\documentclass{llncs}
\usepackage{times}

\usepackage{amsfonts,amsmath,amssymb}
\usepackage{mathtools}
\usepackage{algorithm}
\usepackage{cancel}
\usepackage[noend]{algpseudocode}
\usepackage{pgf}
\usepackage{tikz}
\usetikzlibrary{arrows,automata,calc}

\newcommand{\scell}[2][c]{%
  \begin{tabular}[#1]{@{}c@{}}#2\end{tabular}}

\newcommand{\dotcup}[0]{\,\dot\cup\,}
\newcommand{\Subterms}[0]{\mathit{Subterms}}

\def\int{\mbox{{\it int}}}
\def\cond{{\mbox{{\it Cond}}}}
\def\term{{\mbox{{\it Term}}}}

\def\intconst{{\mbox{{\it Const}}}}
\def\ITE{{\mbox{{\it ITE}}}}

\newcommand{\comment}[1]{}

\newcounter{myctr}

\newenvironment{myitemize}{\begin{list}{$\bullet$}
{\setlength{\topsep}{1mm}\setlength{\itemsep}{0.25mm}
\setlength{\parsep}{0.1mm}
\setlength{\itemindent}{0mm}\setlength{\partopsep}{0mm}
\setlength{\labelwidth}{15mm}
\setlength{\leftmargin}{4mm}}}{\end{list}}

\newcommand{\rig}[0]{\vdash^r}
\newcommand{\ahat}[1]{\expandafter\hat#1}
\newcommand{\restr}[0]{\upharpoonright}

\newcommand{\smvertspace}{\vspace*{-2mm}}

\author{Benjamin Caulfield${}^1$, Markus N. Rabe${}^1$, Sanjit A. Seshia${}^1$, \\ and Stavros Tripakis${}^{1,2}$}
\institute{${}^1$University of California, Berkeley\\${}^2$Aalto University} % Department of Electrical Engineering and Computer Science

\title{What's Decidable about\\ Syntax-Guided Synthesis?}

\begin{document}

\maketitle

\begin{abstract}
	Syntax-guided synthesis (SyGuS) is a recently proposed fra\-mework for program synthesis problems. 
	The SyGuS problem is to find an expression or program generated by a given grammar that meets a correctness specification. 
	Correctness specifications are given as formulas in suitable logical theories, typically amongst those 
        studied in satisfiability modulo theories (SMT). 

	\indent In this work, we analyze the decidability of the SyGuS problem for different classes of grammars and correctness specifications. 
%	The rest of the paper thus focuses on more restricted grammars, such as regular tree grammars. 
	We prove that the SyGuS problem is undecidable for the theory of equality with uninterpreted functions (EUF).
	We identify a fragment of EUF, which we call \emph{regular-EUF}, for which the SyGuS problem is decidable.
	We prove that this restricted problem is EXPTIME-complete and that the sets of solution expressions are precisely the regular tree languages.
	For theories that admit a unique, finite domain, we give a general algorithm to solve the SyGuS problem on tree grammars. 
	Finite-domain theories include the bit-vector theory without concatenation. 
	We prove SyGuS undecidable for a very simple bit-vector theory with concatenation, both for context-free grammars and for tree grammars. 
        Finally, we give some additional results for linear arithmetic and bit-vector arithmetic along with a discussion of the implication of these results.
\end{abstract}

%%%%%%%%%%%%%%%%%%%%%%%%%%%%%%%%%%%%%%%%%%%%%%%%%%%%%%%%%%%%%%%%%%%%%%%%%%%%%%%%%%%%%
\section{Introduction}
\label{sec:intro}

{\em Program synthesis} is an area concerned with the automatic generation of
a program from a high-level specification of correctness. The specification may either
be total, e.g., in the form of a simple but unoptimized program, or partial, e.g., in 
the form of a logical formula or even a collection of test cases. Regardless, one
can typically come up with a suitable logic in which to formally capture the class of 
specifications. 
Traditionally, program synthesis has been viewed as a deductive process, wherein
a program is derived from the constructive proof of the theorem that states that for all inputs,
there exists an output, such that the desired correctness specification holds~\cite{MannaWaldinger80},
with no assumptions made about the syntactic form of the program.
However, over the past decade, there has been a successful trend in synthesis in which, in addition to
the correctness specification, one also supplies a hypothesis about the syntactic 
form of the desired program.
Such a hypothesis can take many forms: partial programs with ``holes''~\cite{sketching:pldi05,sketching:asplos06},
component libraries~\cite{jha-icse10,gulwani-pldi11}, 
protocol scenarios~\cite{udupa-pldi13,ScenariosHVC2014}, etc.
Moreover, the synthesis of verification artifacts, such as invariants~\cite{colon-cav03}, 
also makes use of ``templates'' constraining their syntactic structure. The intuition is that
such syntactic restrictions on the form of the program reduce the search space for 
the synthesis algorithms, and thus speed up the overall synthesis or verification process.

{\em Syntax-guided synthesis} (SyGuS)~\cite{alur-fmcad13} is a recently-proposed formalism
that captures this trend as a new class of problems. More precisely, a SyGuS problem comprises
a logical specification $\varphi$ in a suitable logical theory $T$ 
that references one or more typed function symbols $f$ that
must be synthesized, along with one or more formal languages $L$ of expressions of the same type
as $f$, with the goal of finding expressions $e \in L$ such that when $f$ is replaced
by $e$ in $\varphi$, the resulting formula is valid in $T$.
The formal language $L$ is typically given in the form of a grammar $G$.
Since the SyGuS definition was proposed about three years ago, it has been adopted by
several groups as a unifying formalism for a class of synthesis efforts, with a standardized
language (Synth-LIB) and an associated annual competition.
However, the theoretical study of SyGuS is still in its infancy.
%from a purely theoretical perspective, we are still in the process of
%development a deeper understanding of the SyGuS problem.
Specifically, to our knowledge, there are no published results about the decidability or complexity 
of syntax-guided synthesis for specific logics and grammars.
\begin{table}[b]
\centering
\smvertspace
\begin{tabular}{|c|c|c|c|}
\hline 
 Theory $\backslash$ Grammar Class & Regular Tree & Context-free \\ \hline
 Finite-Domain & D & U \\ \hline 
 Bit-Vectors & U & U \\ \hline 
 Arrays  & U & U \\ \hline 
 EUF  & U & U \\ \hline
 Regular-EUF  & D & ? \\ \hline
\end{tabular}
\vspace{4pt}
\caption{Summary of main results, organized by background theories and classes of grammars. 
``U'' denotes an undecidable SyGuS class, ``D'' denotes a decidable class,
and ``?'' indicates that the decidability is currently unknown.}
\label{tbl:main-results-summary}
\end{table}
In this paper, we present a theoretical analysis of the syntax-guided synthesis problem.
We analyze the decidability of the SyGuS problem for different classes of grammars and logics. 
For grammars, we consider arbitrary context-free grammars, tree grammars, and
grammars specific to linear real arithmetic and linear integer arithmetic. 
For logics, we consider the major theories studied in satisfiability modulo theories (SMT)~\cite{barrett-smtbookch09},
including equality and uninterpreted functions (EUF),
finite-precision bit-vectors (BV), and
arrays -- extensional or otherwise (AR), as well as theories with finite domains (FD).
Our major results are as follows:
\begin{myitemize}
\item For EUF, we show that the SyGuS problem is undecidable over tree grammars. These results extend straightforwardly for 
the theory of arrays. (See Section~\ref{sec:euf}.)
\item We present a fragment of EUF, called \emph{regular-EUF}, for which the SyGuS problem is EXPTIME-complete given regular tree grammars. 
We prove that the sets of solution to regular-EUF problems are in one-to-one correspondence with the regular tree languages.
(See Section~\ref{sec:reg-euf}.)
\item For arbitrary theories with finite domains (FD) defined in Section~\ref{sec:fd}, we show that the SyGuS problem is decidable for tree grammars, but undecidable for arbitrary context-free grammars.
\item For BV, we show (perhaps surprisingly) that the SyGuS problem is undecidable for
the classes of context-free grammars and tree grammars. (See Section~\ref{sec:bv}.)
\end{myitemize}
See Table~\ref{tbl:main-results-summary} for a summary of our main results.
In addition, we also consider certain restricted grammars specific to 
the theory of linear arithmetic over the reals and integers (LRA and LIA),
as well as bit-vectors (BV) where the grammars generate arbitrary but well-formed expressions
in those theories and discuss the decidability of the problem in Section~\ref{sec:misc}.
The paper concludes in Section~\ref{sec:discuss} with a discussion of the results, their implications,
and directions for future work.

%\begin{figure}
%\centering
%\begin{tabular}{|c|c|c|c|}
%\hline 
% Theory $\backslash$ Grammar Class & Regular& Tree & Context-free  \\ \hline
% Finite-Domain & ? & D & U \\ \hline 
% Bit-Vector & ? & ? & U \\ \hline 
% DL & ? & ? & U  \\ \hline
% LRA & ? & ? & U  \\ \hline 
% LIA & ? & ? & U  \\ \hline
% Arrays & U & U & U \\ \hline 
% EUF & U & U & U \\ \hline
%\end{tabular}
%\caption{The results ordered by background theories and classes of grammars.}
%\label{fig:my_label}
%\end{figure}

%Contributions
%\begin{itemize}
%	\item SyGuS for EUF is undecidable, even for severely restricted grammars. 
%	\item SyGuS for LIA and LRA is undecidable for arbitrary context-free grammars. 
%	\item For decidable theories over finite domains (assuming the theories are a fragment of first-order), SyGuS is decidable. 
%\end{itemize}

%Let $\lambda$ represent the empty string. 
%G is a grammar, we represent it by production rules (non-terminals, terminals, etc. are infered by these rules) 

%%%%%%%%%%%%%%%%%%%%%%%%%%%%%%%%%%%%%%%%%%%%%%%%%%%%%%%%%%%%%%%%%%%%%%%%%%%%%%%%%%%%%
\smvertspace
\section{Preliminaries}
\label{sec:prelim}
\smvertspace

%We survey basic definitions and notation on logical theories,
%syntax-guided synthesis, and grammars.
% 
%\subsection{Logical Theories}
%\label{sec:logics}

%\section{Background}

%TAKEN FROM BACKGROUND:
%We review some key definitions of logical theories used in
%satisfiability modulo theories (SMT);
%see~\cite{barrett-smtbookch09,hodges1997shorter} for more details. 
%\cite{hodges1997shorter}
We review some key definitions and results used in the rest of the paper. 

\smvertspace
%{\bf Term Rewriting Systems.}~\cite{baader1999term} 
\subsubsection*{Terms and Substitutions}
We follow the book by Baader and Nipkow~\cite{baader1999term}. 
A \emph{signature} (or \emph{ranked alphabet}) $\Sigma$ consists of a set of \emph{function} symbols with an associated \emph{arity}, a
non-negative number indicating the number of arguments.  
For example $\Sigma = \{ f:2, a:0, b:0\}$ consists of binary function symbol $f$ and constants $a$ and $b$.
For any arity $n \ge 0$, we let $\Sigma^{(n)}$ denote the set of function symbols with arity $n$ (the $n$-ary symbols). 
We will refer to the $0$-ary function symbols as \emph{constants}. 

For any signature $\Sigma$ and set of \emph{variables} $X$ such that $\Sigma \cap X = \emptyset$, we define the set $T(\Sigma, X)$ of $\Sigma$-terms over $X$ inductively as the smallest set satisfying:
\begin{itemize}
\item $\Sigma^{(0)},X \subseteq T(\Sigma, X)$
\item  For all $n \ge 1$, all $f \in \Sigma^{(n)}$, and all $t_1, \dots, t_n \in T(\Sigma, X)$, we have $f(t_1,\dots,t_n) \in  T(\Sigma, X)$.
\end{itemize} 

We define the set of \emph{ground terms} of $\Sigma$ to be the set $T(\Sigma, \emptyset)$ (or short $T(\Sigma)$).
%
%The set of \emph{positions} of a term $t$, denoted $Pos(t)$, is a set of strings over the alphabet of positive integers. It is inductively defined as follows:
%\begin{itemize}
%\item If $t \in X \cup \Sigma^{(0)}$, then $Pos(t) = \{ \epsilon \}$
%\item If $t = f(t_1, \dots, t_n)$, then $Pos(t) = \{\epsilon\}  \cup \bigcup_{i=1}^{n} \{ip \mid p \in Pos(t_i) \}$ 
%\end{itemize}
%
%\ben{we should look at this section after the paper is mostly done. I suspect that many of these definitions will not be needed}
%Here, $\epsilon$ represents the empty string and is called the \emph{root position} of $t$.
%For positions $p$ and $q$, we say $p \le q$ if there exists a position $p'$ such that $pp' = q$.  
%We say $p$ is parallel to $q$, denoted $p \| q$, if $p \not\le q$ and $q \not\le p$. 
%The size of a term $t$, denoted $|t|$, is $|Pos(t)|$. 
%For $p \in Pos(t)$, the subterm of $t$ at position $p$, denoted $t|_p$ is defined by:
%\begin{itemize}
%\item $t|_\epsilon = t$ 
%\item $t|_{ip'} = t_i|_{p'}$, if $t = f(t_1, \dots, t_n)$
%\end{itemize}
%
%For $p \in Pos(t)$, the term $t[s]_p$ is created by replacing the subterm at position $p$ with $s$. In other words,
%\begin{itemize}
%\item $t[s]_\epsilon = s$
%\item $f(t_1, \dots, t_n)[s]_{ip'} = f(t_1, \dots, t_i[s]_{p'}, \dots, t_n)$
%\end{itemize}
%
%We define $\Subterms(t) = \{ t|_p \mid p \in Pos(t) \}$ and we lift the definition to sets $S$ of terms, $\Subterms(S) = \bigcup_{s \in S}\Subterms(s)$. 
%We say that a set $S$ of terms is \emph{subterm-closed} if $\Subterms(S) = S$. 
%
We define the subterms of a term recursively as $\Subterms(g(s_1,\dots,s_k)) = \{g(s_1,\dots,s_k)\} \cup \bigcup_i \Subterms(s_i)$, which we lift to sets $S$ of terms, $\Subterms(S) = \bigcup_{s \in S}\Subterms(s)$. 
We say that a set $S$ of terms is \emph{subterm-closed} if $\Subterms(S) = S$. 

For a set $y_1,\dots,y_k$ of variables (or constants) and terms $t_1,\dots,t_k,s$, the term $s\{t_1 / y_1, \dots, t_k / y_k\}$ is formed by replacing each instance of each $y_i$ in $s$ with $t_i$.
We call $\sigma : = \{t_1 / y_1, \dots, t_k / y_k\}$ a \emph{substitution}.
Substitutions extend in the natural way to formulae, by applying the substitution to each term in the formula.

We extend substitution to function symbols with arity $>0$, where it is also called \emph{second-order} substitution. 
For a function symbol $f$ of arity $k$, a signature $\Sigma$, and a fresh set of variables $\{x_1,\dots,x_k\}$, a \emph{substitution} to $f$ in $\Sigma$ is a term $w \in T(\Sigma, \{x_1,\dots,x_k\})$. 
Given a term $s \in T(\Sigma \cup {f})$, the term $s\{w/f\}$ is formed by replacing each occurrence of any term $f(s_1,\dots,s_k)$ in $s$ with $w\{s_1 / x_1, \dots, s_k / x_k \}$ (sometimes written $w(s_1, \dots, s_k)$).
We say that $x_1,\dots,x_k$ are the \emph{bound variables} of $f$. %^, and will always use these variables in any substitution to a function symbol.
Intuitively, second-order substitution replaces not only $f$ by $w$, but also replaces the arguments $s_1,\dots,s_k$ of each function application $f(s_1,\dots,s_k)$ by the bound variables. 

A \emph{context} $B$ is a term in $T(\Sigma, \{x\})$ with a single occurrence of $x$. 
For $s \in T(\Sigma)$, we write $B[s]$ for $B\{s/x\}$. 
%\ben{should we bother defining $\rightarrow_E$? It seems like everything can be done using just equational logic.}
%Given a set of equations $E \subseteq T(\Sigma) \times T(\Sigma)$ and terms $s,t \in T(\Sigma)$, we say that $s \rightarrow_E t$ if there exists an $(l,r)$ in $E$ and a $p \in Pos(s)$ such that $s|_p = l$ and $s[r]_p = t$. 
%Let $\estar{E}$ be the symmetric and transitive closure of $\rightarrow_E$. 
%
%We say a relation $\equiv \subseteq T(\Sigma) \times T(\Sigma)$ is  closed under $\Sigma$-operations if $s_1 \equiv t_1 \dots s_n \equiv t_n$ implies that $g(s_1 \dots s_n) = g(t_1 \dots t_n)$ for all $n > 0$, all terms $s_1,t_1\dots s_n, t_n \in T(\Sigma)$ and all $g \in \Sigma^(n)$.
%For any set of equation $E \subset T(\Sigma) \times T(\Sigma)$, $\estar{E}$ is the smallest equivalence relation containing $E$ that is closed under $\Sigma$-operations \cite{baader1999term}. 

\smvertspace
\subsubsection*{Logical Theories} 
A first-order model $M$ in $\Sigma$, also called $\Sigma$-model, is a
pair consisting of a set $\textit{dom}(M)$ called its domain and a
mapping $(-)^M$.  
The mapping assigns to each function symbol $f \in \Sigma^F$ with
arity $n\geq 0$ a total function $f^M:\textit{dom}(M)^n\to
\textit{dom}(M)$, and to each relation $R \in \Sigma^R$ of arity $n$ a
set $R^M \subseteq \textit{dom}(M)^n$.

A \emph{formula} is a boolean combination of relations over terms. 
The mapping induced by a model $M$ defines a natural mapping of formulas 
$\varphi\in L(\Sigma)$ to truth values, written $M \models \varphi$ (we also say $M$ \emph{satisfies} $\varphi$).  
For some set $\Phi$ of first-order formulas, we say $M \models \Phi$
if $M \models \varphi$ for each $\varphi \in \Phi$. 
A \emph{theory} $\mathcal{T}\subseteq L (\Sigma^F\cup\Sigma^R)$ is a set of formulas. 
We say $M$ is a model of $\mathcal{T}$ if $M \models \mathcal{T}$, and
use $\text{Mod}(\mathcal{T})$ to denote the set of models of
$\mathcal{T}$.  
A first-order formula $\varphi$ is \emph{valid} in $\mathcal{T}$ if
for all $M \in \text{Mod}(\mathcal{T})$, $M \models \varphi$. 
%A \emph{sentence} is a formula without free variables. 
A theory is \emph{complete} if for all formulas $\varphi\in L(\Sigma)$ either $\varphi$ or $\neg\varphi$ is valid.

Given a set of ground equations $E \subseteq T(\Sigma) \times T(\Sigma)$ and terms $s,t \in T(\Sigma)$, we say that $s \rightarrow_E t$ if there exists an $(l,r)$ in $E$ and a context $C$ such that $C[l] = s$ and $C[r] = t$.
For example, if $E := \{ a = g(b) \}$, then $h(a) \rightarrow_E h(g(b))$. 
%$p \in Pos(s)$ such that $s|_p = l$ and $s[r]_p = t$. 
Let $=_E$ be the symmetric and transitive closure of $\rightarrow_E$. 
We will sometimes write $E \vdash s=t$ instead of $s=_E t$.
We will use $[s]_E$ to represent the set $\{ t \mid s=_E t\}$. 
\emph{Birkhoff's Theorem} states that for any ranked alphabet $\Sigma$, set $E \subseteq T(\Sigma) \times T(\Sigma)$ and $s,t \in T(\Sigma)$,  $E \vdash s=t$ if and only if for every model $M$ in $\Sigma$ such that $M \models \bigwedge_{e\in E} e$ it holds $M \models s=t$~\cite{baader1999term}.

%The equations provable from a set $E$ of equations can also be derived using \

In this work, we consider the common quantifier-free background
theories of SMT solving: propositional logic (SAT), bit-vectors (BV),
difference logic (DL), linear real arithmetic (LRA), linear integer
(Presburger) arithmetic (LIA), the theory of arrays (AR), and the
theory of uninterpreted functions with equality (EUF). For detailed
definitions of these theories, see~\cite{barrett-smtbookch09,BarFT-SMTLIB}.
%\markus{Should we give detailed definitions of their signatures?}

%For the theory of EUF, in particular, we assume the presence of an If-Then-Else ($\ITE$) operator, which takes inputs of the form $\ITE(formula, term, term)$.
%In particular, for terms $s$ and $t$, $\ITE(true,s,t) = s$ and $\ITE(false,s,t)=t$ \cite{bryant1999exploiting}. 

For the theory of EUF it is common to introduce the If-Then-Else operator ($\ITE$) as syntactic sugar~\cite{bryant1999exploiting,barrett-smtbookch09,BarFT-SMTLIB}. 
We follow this tradition and allow EUF formulas to contain terms of the form $\ITE(\varphi,t_1,t_2)$, where $\varphi$ is a formula, and $t_1$ and $t_2$ are terms. 
To desugar EUF formulas we introduce an additional constant $c_{\text{ite}}$ and add two constrains $\varphi \rightarrow (c_{\text{ite}} = t_1)$ and $\neg \varphi \rightarrow (c_{\text{ite}} = t_2)$ for each $\ITE$ term $\ITE(\varphi,t_1,t_2)$.
As we will see in Section~\ref{sec:euf} the presence syntactic sugar such as the $\ITE$ operator in the grammar of SyGuS problems may have a surprising effect on the decidability of the SyGuS problem. 
%\subsection{Grammars and Automata} %\ben{do we really need to discuss CFGs?}
%\label{sec:grammars}
%\noindent
%{\bf{Grammars and Automata.}} %\ben{do we really need to discuss CFGs?}

\smvertspace
\subsubsection*{Grammars and Automata} %\ben{do we really need to discuss CFGs?}
\label{sec:grammars}
A \emph{context-free grammar} (CFG) is a tuple $G = (N, T, S, R)$ 
consisting of a finite set $N$ of nonterminal symbols 
with a distinguished \emph{start symbol} $S\in N$, 
a finite set $T$ of terminal symbols, 
and a finite set $R$ of production rules, which are tuples of the form $(N,(N\cup T)^*)$. 
Production rules indicate the allowed replacements of non-terminals by
sequences over nonterminals and terminals.  
The \emph{language}, $L(G)$, generated by a context-free grammar is
the set of all sequences that contain only terminal symbols that can
be derived from the start symbol using the production rules.  

\emph{Tree grammars} are a more restrictive class of grammars.
They are defined relative to a \emph{ranked alphabet} $\Sigma$.
%This is mentioned above.
%We will often refer to a ranked alphabet by its set of symbols, $\Sigma$ and assume that the arity is defined. 
%We define the set $T(\Sigma)$ of \emph{$\Sigma$-trees} as
%the smallest set that includes the symbols $g\in\Sigma$ with
%arity 0 and further includes for all $g \in \Sigma$ with arity $k$
%and elements $t_1,\dots,t_k\in T_{\Sigma}$ the element
%$g(t_1,\dots,t_k)$. 
A \emph{regular tree grammar} $G = (N,S,\Sigma,R)$ consists of a set $N$
of non-terminals, a start symbol $S \in N$, a ranked alphabet
$\Sigma$, and a set $R$ of production rules.  
Production rules are of the form $A \rightarrow g(t_1,t_2,...,t_k)$,
where $A\in N$, $g$ is in $\Sigma$ and has arity $k$, and each
$t_i$ is in $N\cup T_{\Sigma}$.  
%We will simply 
For a given tree-grammar $G$ we write $L(G)$ for the set of trees produced by $G$.
The \emph{regular tree languages} are the languages produced by some regular tree grammar.
Any regular tree grammar can be converted to a CFG by simply treating the
right-hand side of any production as a string, rather than a tree. 
Thus, the {\em undecidability results for SyGuS given regular tree grammars extend to undecidability results for SyGuS given CFGs}. 
%We will use $L(G)$ also to denote the set of strings produced by %\ben{We don't seem to do this anywhere in the paper}
%the corresponding CFG. 

Let $\Sigma$ be a signature of a background theory $\mathcal{T}$.  
We define a tree grammar $G=(N,S,\Sigma,P)$ to be
\emph{$\mathcal{T}$-compatible} (or $\Sigma$-compatible) if $\Sigma \subseteq \Sigma^F
\cup \Sigma^R$ and the arities for all symbols in $\Sigma$ match
those in $\Sigma$. 

A (deterministic) bottom-up (or \emph{rational}) tree automaton $A$ is a tuple $(Q, \Sigma, \delta, Q_F)$. 
Here, $Q$ is a set of states, $Q_F \subseteq Q$, and $\Sigma$ is a ranked alphabet.
The function $\delta$ maps a symbol $g \in \Sigma^{(k)}$ and states $q_1,\dots,q_k$ to a new state $q'$, for all $k$.
If no such $q'$ exists, $\delta(g, q_1,\dots,q_k)$ is undefined.
We can inductively extend $\delta$ to terms, where for all $g \in \Sigma^{(k)}$ and all $s_1,\dots,s_k \in T(\Sigma)$, we set $\delta(g(s_1,\dots,s_k)) := \delta(g, \delta(s_1), \dots, \delta(s_k))$.
The language accepted by $A$ is the set $L(A) := \{ s \mid s\in T(\Sigma),  \delta(s)\in Q_F\}$.
There exist transformations between regular tree grammars and rational
tree automata \cite{tata2007}, and we will sometimes also define SyGuS problems in terms of rational tree automata rather than a regular tree grammars.
\smvertspace
\subsubsection*{Syntax-Guided Synthesis}
\label{sec:sygus}
%We follow the definition of SyGuS given by Alur et al.~\cite{alur-fmcad13} and focus on the case that the set of candidate expressions (programs) is generated by a given grammar.  
%Let $\mathcal{T}$ be a background theory, and let $\mathcal{G}$ be a class of grammars. 
%Given a function symbol $f$ with arity $r$, a formula $\varphi$ over the signature $\Sigma$ of $\mathcal{T}$ along with $f$, and a grammar $G\in\mathcal G$ of expressions in $\Sigma \cup \{x_1,\dots,x_{r}\}$, the SyGuS problem is to find an expression $e \in L(G)$ such that the formula $\varphi\{e/f\}$ is valid or to determine the absence of such an expression. 
%\markus{What are expressions? Also: $\varphi\{e/f\}$ is not in $\mathcal T$ because of the variables $x_1,\dots,x_r$. I think what we need is applying two substitutions $\varphi\{e/f(e_1,\dots,e_r)\}\{e_1/x_1,\dots,e_r/x_r\}$ for each application $f(e_1,\dots,e_r)$.}
%We represent the SyGuS problem as the tuple $(\varphi, \Sigma, G, f)$.
%
We follow the definition of SyGuS given by Alur et al.~\cite{alur-fmcad13}, but we focus on the case to find a replacement for a single designated function symbol $f$ with a candidate expression (the program), which is generated by a given grammar $G$. 
Let $\mathcal{T}$ be a background theory over signature $\Sigma$, and let $\mathcal{G}$ be a class of grammars. 
Given a function symbol $f$ with arity $k$, a formula $\varphi$ over
the signature $\Sigma\dot\cup \{f\}$, and a grammar $G\in\mathcal G$
of terms in $T(\Sigma,\{x_1,\dots,x_k\})$, the SyGuS problem is to
find a term $w \in L(G)$ such that the formula $\varphi\{w/f\}$ is
valid or to determine the absence of such a term.  
We represent the SyGuS problem as the tuple $(\varphi, {\mathcal T}, G, f)$. 

The variables $x_1,\dots,x_k$ that may occur in the generated term $w$ stand for the $k$ arguments of $f$. 
For each function application of $f$ the higher-order substitution $\varphi\{w/f\}$ then replaces $x_1,\dots,x_k$ by the arguments of the function application. 

%Accordingly, the candidate expressions are terms in $T(\Sigma,\{x_1,\dots,x_k\})$, where $x_1,\dots,x_k$ are variables standing for the $k$ arguments of $f$ (which will serve as the bound variables of $f$ in the higher-order substitution as defined above). 

Note that the original definition of SyGuS allows for universally quantified variables, while our definition above admits no variables. 
This is equivalent as universally quantified variables can be replaced with fresh constants without affecting validity. 
%For example $\forall x, y f(x) = g(y)$ is valid if and only if $f(c_x) = g(c_y)$ is valid for new constants $c_x$ and $c_y$.
%Therefore, we will sometimes assume without loss of generality that the constraint formula contains no variables. 

%That is, the SyGuS problem is parametric in the background theory $T$ and the class of grammars $\mathcal G$. 
%
%Even though the problem is given in terms of
%synthesizing a function $f$, usually the grammar generates expressions
%that are {\it applications} of $f$. \ben{what does this mean?} % on terms constructed from symbols 
%from the logical formula and the underlying logical theories. Thus, the 
%expressions $e$ replace applications of $f$ in $\varphi$.
\begin{example}
Consider the following example SyGuS problem in linear integer arithmetic.
Let the type of the function to synthesize $f$ be 
$\int\times\int\mapsto\int$ and let the specification be given by the logical formula 
\[\varphi_1:\ \forall x, y \; f(x,y)=f(y,x)\ \wedge\ f(x,y)\ge x\,.\]
We can restrict the set of expressions
$f(x,y)$ to be expressions 
generated by the grammar below:
\begin{eqnarray*}
\term & := & x\mid y\mid\intconst\mid \ITE(\cond,\term,\term)\\
\cond & := & \term\le\term\mid \cond\wedge\cond \mid \neg\cond\mid (\cond)
\end{eqnarray*}
It is easy to see that a function computing the maximum over $x$ and $y$, such as $\ITE(x \le y, y, x)$, is a solution to the SyGuS problem. 
There are, however, other solutions, such as $\ITE(7 \le y \,\wedge\, 7 \le x, \,\ITE(x \le y, y, x), 10)$. 
The function computing the sum of $x$ and $y$ would satisfy the specification, but cannot be constructed in the grammar. 
%\qed
\end{example}

%%%%%%%%%%%%%%%%%%%%%%%%%%%%%%%%%%%%%%%%%%%%%%%%%%%%%%%%%%%%%%%%%%%%%%%%%%%%%%%%%%%%%

%\section{EUF is Decidable}
%\label{sec:euf_dec}

%\input{EUF_is_decidable}

%%%%%%%%%%%%%%%%%%%%%%%%%%%%%%%%%%%%%%%%%%%%%%%%%%%%%%%%%%%%%%%%%%%%%%%%%%%%%%%%%%%%%

\smvertspace
\section{SyGuS-EUF is Undecidable}
\label{sec:euf}
%\label{sec:euf_dec}
\smvertspace

We use SyGuS-EUF to denote the class of SyGuS problems $(\varphi, \text{EUF}, G, f)$ where $G$ is a grammar generating expressions that are syntactically well-formed expressions in EUF for $f$. 
In this section, we prove that SyGuS-EUF is undecidable. 
The proof of undecidability is a reduction from the simultaneous rigid E-unification problem (SREU) \cite{degtyarev1996undecidability}.
We say that a set $E:= \{e_1,\dots,e_l\}$ of equations between terms in $T(\Sigma, V)$ together with an equation $e^*$ between terms in $T(\Sigma, V)$ forms a \emph{rigid expression}, denoted $E \rig e^*$.  
A \emph{solution} to $E \rig e^*$ is a substitution $\sigma$, such that $e^*\sigma$ and $e_i \sigma$ are ground for each $e_i \in E$ and $E  \sigma \vdash e^*\sigma$. 
Given a set $S$ of rigid equations, the SREU problem is to find a substitution $\sigma$ that is a solution to each rigid equation in $S$, and is known to be undecidable~\cite{degtyarev1996undecidability}. 

\noindent
{\bf Reducing SREU to SyGuS-EUF.} 
We start the reduction with constructing a boolean expression $\Phi_S$ for a given set of rigid equations $S$ over alphabet $\Sigma$ and variables $V := \{x_1,\dots,x_m\}$. 
Let each $r_i\in S$ be $e_{i,1},\dots,e_{i,l_i} \rig e^*_i$, where $e_{i,1},\dots,e_{i,l_i}$, and $e^*_i$ are equations between terms in $T(\Sigma, V)$. 
We associate with each rigid expression $r_i\in S$ a boolean expression $\psi_i := \big(\bigwedge_{j = {1,\dots,l_i}} {e}_{i,j}\sigma_f \wedge \bigwedge_{k \ne j} a_k \ne a_j\big) \rightarrow {e_i}^*\sigma_f$, where $\sigma_f$ is the substitution $\{ f(a_1) / x_1, \dots, f(a_m) / x_m\}$. 
The symbol $f$ is a unary function symbol to be synthesized and $a_1, \dots, a_m$ are fresh constants ($a_i\notin\Sigma$ for all $i$).
We set $\Phi_S := \bigwedge_i \psi_i$. 

Next we give the grammar $G_S$, which generates the terms that may replace $f$ in $\Phi_S$. 
We define $G_S$ to have the starting nonterminal $A_1$ and the following rules:
\[
\begin{array}{ll}
A_1 \rightarrow & ITE(x = a_1, S', A_{2}) \\
A_2 \rightarrow & ITE(x = a_2, S', A_{3}) \\
\dots \\ 
A_{m-1} \rightarrow & ITE(x = a_{m-1}, S', A_{m-1}) \\
A_{m} \rightarrow & ITE(x = a_m, S', \bot) \\
%\{ A_i \rightarrow ITE(x = a_i, S', A_{i+1}) \; | \; i \in \{ 1, \dots, m\} \} \\ 
\end{array}
\]
where $\bot$ is a fresh constant ($\bot\notin \Sigma$ and $\bot\neq a_i$ for all $i$). 
Additionally, for each $g\in\Sigma$ we add a rule $S' \rightarrow g(S',\dots,S')$, where the number of argument terms of $g$ matches its arity. 

\begin{lemma}
\label{lem:sreutoeuf}
The SREU problem $S$ has a solution if and only if the SyGuS-EUF problem $\rho_S := (\Phi_S, \text{EUF}, G_S, f)$ has a solution over the ranked alphabet $\Sigma$. 
\end{lemma}
\begin{proof}
The main idea behind this proof is that each $f(a_i)$ in $\Phi_S$ represents the variable $x_i$ in $S$. 
Any replacement to $f$ found in $G_S$ corresponds to a substitution on all variables $x_i$ in $S$ that grounds the equations in the SREU problem.

$\rightarrow:$ Let $\sigma_u := \{ u_1 / x_1 \dots, u_m / x_m\}$ be a solution to $S$, where each $u_i$ is a ground term in $T(\Sigma)$. 
We consider the term $w(x) := ITE(x = a_1, u_1, ITE(x = a_2, u_2, \dots,$ $ITE(x=a_{m}, u_m, \bot)\dots)$, which is in the language of the grammar $G_S$. 
To show that $\Phi_s\{w / f\}$ is valid, it suffices to show that for each model $M$ of $\Sigma \cup \{a_1, \dots, a_m\} \cup V$ and for each $\psi_i$ we have $M \models \psi_i\{w / f\}$. 
%\ben{There's a subtlety here as to whether models can give interpretations of universally quantified variables. It won't affect the result, but I should address it somewhere}
If $M \not\models [\bigwedge_{j = {1,\dots,l_i}} {e}_{i,j}\sigma_f \wedge \bigwedge_{k \ne j} a_k \ne a_j]\{w/f\}$, then $M \models \psi_i\{w / f\}$ holds trivially. 
We handle the remaining case below, giving justifications to the right of each new equation.
\begin{enumerate}
 \item Assume $M \models [\bigwedge_{j = {1,\dots,l_i}} {e}_{i,j}\sigma_f \wedge \bigwedge_{k \ne j} a_k \ne a_j]\{w/f\}$ 
 \item $M \models  \bigwedge_{k \ne j} a_k \ne a_j$ \hfill (1)
 \item For each $j$: $M \models w(a_j) = u_j$ \hfill (2)
 \item For each $j$: $M \models ({e}_{i,j}\sigma_f)\{w/f\} \leftrightarrow e_{i,j}\sigma_u$ \hfill (3)
 \item $M \models \bigwedge_{j = {1,\dots,l_i}} ({e}_{i,j}\sigma_f)\{w/f\}$ \hfill (1)
 \item $M \models \bigwedge_{j = {1,\dots,l_i}} e_{i,j}\sigma_u$ \hfill (4, 5)
 \item $\{e_{i,j} \mid j=1,\dots,m \}\sigma_u \vdash e^*_i\sigma_u$ \hfill (def. SREU)
 \item $M \models e^*\sigma_u$ \hfill (6,7, Birkhoff's Thm.)
 \item $M \models (e^*\sigma_f) \{w/f\}$ \hfill (3,8)
% \item $M \models  \ahat{s}_{i,j}\{w/f\} = s_{i,j}\sigma_u \wedge \ahat{t}_{i,j}\{w/f\} = t_{i,j}\sigma_u$ \hfill (3) 
% \item $M \models \ahat{s}_{i,j}\{w/f\} = \ahat{t}_{i,j}\{w/f\}$ \hfill (8,9)
\end{enumerate}
Therefore, $M \models \Phi_S$ and we get that $w$ is a solution to the SyGuS problem $\rho_S$.

$\leftarrow:$ Let $w(x)$ and $\sigma_u$ be defined as before and assume that $w$ is a solution to the SyGuS problem $\rho_S$. 
Each $u_i$ in $w$ is ground, since the nonterminal $S'$ in $G_S$ can only produce ground terms.
Chose any $r_i \in S$. 
We will show for every model $M$ on $\Sigma \cup V$, that if $M \models  \bigwedge_{j = {1,\dots,l_i}} e_{i,j}\sigma_u$ then $M \models  e^*_i\sigma_u$.
By Birkhoff's theorem, this implies $e_{i,1}\sigma_u,\dots,e_{i,l_i}\sigma_u \vdash e^*_i\sigma_u$.
\begin{enumerate}
\item Assume $M \models  \bigwedge_{j = {1,\dots,l_i}} e_{i,j}\sigma_u$
\item Let $\ahat{M}$ be a model over $\Sigma \cup V \cup \{a_1,\dots,a_m\}$ such that $\ahat{M} \restr \Sigma \cup V = M$ and $\ahat{M}$ assigns each $a_i$ to a distinct new element not in $dom(M)$.
\item $\ahat{M} \models w(a_j) = u_j$ \hfill (2)
\item For each $j$: $\ahat{M} \models ({e}_{i,j}\sigma_f)\{w/f\} \leftrightarrow e_{i,j}\sigma_u$  \hfill (3)
\item $\ahat{M} \models  \bigwedge_{j = {1,\dots,l_i}} e_{i,j}\sigma_u$ \hfill (1,2)
\item $\ahat{M} \models \bigwedge_{j = {1,\dots,l_i}} ({e}_{i,j}\sigma_f)\{w/f\} $ \hfill (4,5)
\item $\ahat{M} \models \psi_i\{w/f\}$ \hfill ($w$ is a SyGuS solution)
\item $\ahat{M} \models (e^*_i\sigma_f)\{w/f\}$ \hfill (6, 7)
\item $\ahat{M} \models e^*_i\sigma_u$ \hfill (3,8)
%\item $\ahat{M} \models (\ahat{s_i}\{w/f\} = s_i\sigma_u) \wedge (\ahat{t_i}\{w/f\} = t_i\sigma_u)$ \hfill (3)
%\item $\ahat{M} \models  s_i\sigma_u = t_i\sigma_u$ \hfill (8,9)
\item $M \models e^*_i\sigma_u$ \hfill (2,9)
\end{enumerate}
Thus $e_{i,1}\sigma_u,\dots,e_{i,l_i}\sigma_u \vdash e^*_i\sigma_u$ and $\sigma_u$ is a solution to $S$.
\qed
\end{proof}

%The main idea behind the proof of Lemma~\ref{lem:sreutoeuf} is that each $f(a_i)$ in $\Phi_S$ represents the variable $x_i$ in $S$. 
%Any replacement to $f$ found in $G_S$ corresponds to a substitution on all variables $x_i$ in $S$ that grounds the equations in the SREU problem.
%This immediately yields the following theorem. 
%The full proof can be found in the appendix. 

\begin{theorem}
The SyGuS-EUF problem is undecidable.
\end{theorem}

%\subsection{EUF without ITE}
\noindent
{\bf{Remark on EUF without ITE.}}
A key step in the proof of Lemma~\ref{lem:sreutoeuf} is the use of $\ITE$ statements to allow a single expression $w$ to encode instantiations of multiple different variables.
As discussed in Section~\ref{sec:prelim}, $\ITE$ statements are commonly part of EUF, but some definitions of EUF do not allow for $\ITE$ statements~\cite{Kroening2008}. 
While this syntactic sugar has no effect on the complexity of the validity of EUF formulas, the undecidability of SyGuS-EUF may depend on the availability of $\ITE$ operators.
It remains open whether there exist alternative proofs of undecidability that do not rely on $\ITE$ statements. 

We use SyGuS-Arrays to denote the class of SyGuS problems $(\varphi, Arrays, G, f)$, where
 Arrays is the theory of arrays \cite{barrett-smtbookch09}, and $G$ is a grammar such that $L(G)$ are syntactically well-formed expressions in Arrays for $f$. There is a standard construction for representing uninterpreted functions as read-only arrays~\cite{barrett-smtbookch09}.
Therefore, the undecidability of SyGuS-Arrays follows from the undecidability of SyGuS-EUF, as we state below. 

\begin{corollary}
The SyGuS-Arrays problem is undecidable.
\end{corollary}

%\begin{theorem}
%The SyGuS problem to synthesize two functions is undecidable on EUF without $\ITE$ statements.
%\end{theorem}
%\ben{Markus: do you want to talk about your insight from trying to remove ITE statements?}

%\begin{itemize}
%	\item EUF is syntactically not the same as equational logic (EL)
%	\item The result does not carry over. 
%	\item EL is undecidable if we consider an extension of SyGuS to synthesize multiple function symbols at once. 
%	\item SyGuS-EL with one function symbol is open.
%	\item This is an example of the (unsurprising) sensitivity of the SyGuS problem in syntactical details. 
%\end{itemize}

%%%%%%%%%%%%%%%%%%%%%%%%%%%%%%%%%%%%%%%%%%%%%%%%%%%%%%%%%%%%%%%%%%%%%%%%%%%%%%%%%%%%%

\smvertspace

\section{Regular SyGuS-EUF}
\label{sec:reg-euf}
\smvertspace

This section describes a fragment of \emph{EUF}, which we call \emph{regular-EUF}, for which the SyGuS problem is decidable.
%The \emph{regular-EUF} 

 %\ben{Is this a good name? I call it this because it precisely describes the SyGuS-EUF problems whose set of solutions are regular languages.}
%This subproblem does not allow for the presence of If-Then-Else (ITE) statements, and requires that at most one $f$ appear in any equation.

\begin{definition}
We call $(\phi,\text{EUF},G,f)$ a \emph{regular SyGuS-EUF problem} if $G$ contains no $\ITE$ expressions 
and $\phi$ is a \emph{regular-EUF formula} as defined below. 
%We call $(\phi,\text{EUF},G,f)$ a \emph{regular SyGuS-EUF problem} if $G$ contains no $\ITE$ statements and $\phi$ is a \emph{regular-EUF formula} as defined below. 

A \emph{regular-EUF} formula is a formula $\phi := \bigwedge_i \psi_i$ over some ranked alphabet $\Sigma$, where each $\psi_i$ satisfies the following conditions:
\begin{enumerate}
\item It is a disjunction of equations or the negation of equations.
\item It does not contain any ITE expressions.
\item It contains at most one occurrence of $f$ per equation.
\item It satisfies one of the following cases:
	\begin{itemize}
	\item Case 1: The symbol $f$ only occurs in positive equations.
	\item Case 2: The symbol $f$ occurs in exactly one negative equation, and nowhere else.
	\end{itemize}
\end{enumerate}

We define any disjunction $\psi$ that satisfies the above conditions as \emph{regular}.
We will refer to a regular $\psi$ as case-1 or case-2, depending on which of the above cases is satisfied.
Note that every regular-EUF formula is in conjunctive normal form.

\end{definition}

We will show that for every regular $\psi_i$, we can construct a regular tree automaton $A_{\psi_i}$ accepting precisely the solutions to the SyGuS-EUF problem on $\psi_i$. 
The set of solutions to $\phi$ then becomes $L(G) \cap \bigcap_i L(A_{\psi_i})$, where $G$ is the grammar of possible replacements.
The grammar $G$ can be represented as a deterministic bottom-up tree automaton $A_G$ whose size is exponential in $|G|$ \cite{tata2007}.
The product-automaton construction can be used to determine if $L(G) \cap \bigcap_i L(A_{\psi_i})$ is non-empty, which would imply that a solution exists to the corresponding SyGuS problem. 
This construction takes $O(|A_G|\cdot \prod_i |A_{\psi_i}|)$ time and space.
Note that this is at most exponential even when some of the automata have size exponential in $|\phi|$ or $|G|$.

The connection between sets of ground equations and regular tree languages was first observed by Kozen \cite{kozen1977complexity}, who showed that a language $L$ is regular if and only if there exist a set $E$ of ground equations and collection $S$ of ground terms such that $L = \bigcup_{s \in S} [s]_E$. 
The following, very similar theorem shows that a certain set of equivalence classes of a ground equational theory can be represented by a regular tree automaton.

\begin{theorem} 
\label{aut_thm}
Let $E$ be a set of ground equations over the alphabet $\Sigma$, and let $C$ be a subterm-closed set of terms such that every term in $E$ is in $C$.
There exists a regular tree automaton without accepting states $A_{E,C} := (Q, \Sigma, \delta)$ such that a state in $Q$ represents an equivalence class of a term in $C$.
More formally,  this means that for all terms $s,t \in T(\Sigma)$ such that there exist terms $s',t' \in C$ so that $s =_E s'$ and $t =_E t'$, it holds that $s =_E t$ if and only if $\delta(s) = \delta(t)$. 
\end{theorem}
\begin{proof}

Let $Q := \{ q_s \mid s \in C\}$.
For each term $g(s_1,\dots,s_k) \in C$, for $g \in \Sigma_k$, let $\delta(g, q_{s_1}, \dots, q_{s_k}) = q_{g(s_1,\dots,s_k)}$.

We define the function $\mathit{merge}(q,q')$ to operate on $A_{E,C}$ as follows:
First, remove $q'$ from $q$. For all $q_1, \dots, q_{i-1}, q_{i+1},\dots,q_k, q''$ and $g$ such that $\delta(g, q_1, \dots, q_{i-1}, q', q_{i+1},$ $ \dots, q_k)=q''$, add $\delta(g, q_1, \dots, q_{i-1}, q, q_{i+1}, \dots, q_k)=q''$ to $\delta$. 
If there already exists some $q'''$ such that $\delta(g, q_1, \dots, q_{i-1}, q', q_{i+1}, \dots, q_k) = q'''$, then $\mathit{merge}(q'', q''')$. 

Now for each $s = t$ in $E$, call $\mathit{merge}(q_s, q_t)$.
A simple inductive argument will show that the resulting automaton is $A_{E,C}$.
\qed
%\ben{I'm not sure how much detail to go into here to show the property holds. This is the sort of thing that can be handled by a cumbersome, but easy inductive argument.}
\end{proof}

Let $\psi := e_1 \vee e_2 \vee \dots e_{k-1} \vee \neg e_k \vee \dots \vee \neg e_{k+r}$ be a regular formula. 
Let $P := \{ e_1, \dots, e_{k-1}\} $ and $N := \{ e_k, \dots, e_{k+r} \}$.
We can rewrite $\psi$ to the normal form $\psi := (\bigwedge_{e \in N} e) \rightarrow (\bigvee_{e \in P} e)$.
Solving the SyGuS problem for $\psi_i$ then becomes a problem of finding a $w$ such that $N\{w/f\} \vdash e\{w/f\}$ for some $e\in P$.
The technique to form the automaton $A_{\psi_i}$ that represents the solutions to $\psi_i$ depends on whether $\psi_i$ is case-1 or case-2.

Assume that $\psi$ is case-1 and chose some $s = t \in P$. 
Assume $f$ is not in $s=t$. 
If $N \vdash s = t$, then $\psi_i$ is trivially solvable.
If $N \not\vdash s=t$, then $s=t$ can be removed from $\psi_i$ to yield an equally solvable formula.  
Now assume $f$ is in $s=t$.
Without loss of generality, there is a context $B$ and a set of terms $s_1, \dots, s_{arity(f)}$ such that $s = B[f(s_1, \dots, s_{arity(f)})]$. 
Let $C := \Subterms(N) \cup \Subterms(\{s_1, \dots,  s_{arity(f)}\})$ and let $A_{P,C} := (Q, \Sigma, \delta)$ be the automaton defined in the proof of theorem $\ref{aut_thm}$. 
For each $q \in Q$, there is a ground term $u_q$ such that $\delta(u_q) = q$. 
Let $Q' \subseteq Q$ be the set of states $q$ such that $\delta(B[u_q]) = \delta(t)$. 
By theorem \ref{aut_thm}, $P \vdash B[u_q] = t$ if and only if $q \in Q'$.
Therefore, for any replacement,  $w$, of $f$, $P \vdash (s=t)\{w/f\}$ if and only if $\delta(w(s_1,\dots,s_{arity(f)})) \in Q'$. %\ben{maybe go into more detail, here}

Let $A_{s=t} := (Q, \Sigma \cup \{x_1, \dots, x_{arity(f)}\}, \delta', Q')$ be a tree automaton with accepting states $Q'$. 
For each $x_i$, let $\delta'(x_i) := \delta(s_i)$. 
For all $u \in T(\Sigma)$, let $\delta'(u) := \delta(u)$. 
A simple inductive argument will show that for any replacement $w$ of $f$,  $\delta(w(s_1,\dots,s_{arity(f)})) = \delta'(w(x_1,\dots,x_{arity(f)}))$.
Thus, $L(A_{s=t})$ defines the precise set of terms $w$ such that $P \vdash (s=t)\{w/f\}$.

The set of solutions to $\psi$ can be given by the automaton $A_{\psi_i}$ whose language is $\bigcup_{s=t \in P} L(A_{s=t})$.
This can be found in time and space exponential in $|N|$ using the product construction for tree automata \cite{tata2007}. %\ben{can probably do better than exponential space, but I can't find a result on this}

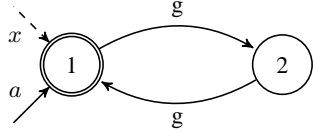
\begin{figure}[!ht]
 \label{ex1fig}
  \centering
  \begin{tikzpicture}[->,>=stealth',shorten >=1pt,auto,node distance=2.8cm,
                    semithick]
  \tikzstyle{every state}=[]

  \node[state,accepting] (A)            {1};
  \node[state]         (B) [right of=A] {2};
  
%  \node[state]         (D) [below right of=A] {$q_d$};
%  \node[state]         (C) [below right of=B] {$q_c$};
%  \node[state]         (E) [below of=D]       {$q_e$};

  \path (A) edge [bend left]  node {g} (B)
  		(B) edge [bend left]  node {g} (A);
  \draw[<-,dashed] (A) -- node[below left] {$x$} ++(-.8cm,.8cm);
  \draw[<-] (A) -- node[above left] {$a$} ++(-.8cm,-.8cm);
\end{tikzpicture}
  \caption{The automaton $A_1$ accepting the solutions to $\psi_1$ in example \ref{regex1}.}
\end{figure}

\begin{example}
\label{regex1}
 Let $\psi := (g(a)=b \wedge g(b)=a) \rightarrow f(a)=g(g(b))$. 
Note that this is a case-1 regular EUF clause. 
If we set $E := \{g(a)=b, g(b)=a\}$ and $C := \{a, b, g(a), g(b), g(g(b))\}$, then $A := A_{E,C}$ is the automaton from figure \ref{ex1fig} (excluding the accepting state and $x$ transition).
Since the argument of $f$ in $f(a)=g(g(b))$ is $a$ and $A$ parses $a$ to state-1, a transition from $x$ to state-1 is added to $A$.
Since $g(g(b))$ parses to state-2 in $A$, state-2 is set as an accepting state in $A$.
So $A$ accepts the replacements $w$ to $f$ such that $\psi\{w/f\}$ is valid.
\end{example}

Assume $\psi$ is case-2 and let $s=t$ be the equation in $N$ that contains $f$. 
Without loss of generality, there is a context $B$ and a set of terms $s_1, \dots, s_{arity(f)}$ such that $s = B[f(s_1, \dots, s_{arity(f)})]$.
Let $N' := N \backslash \{s=t\}$, and let $C := \Subterms(N' \cup P) \cup \Subterms(\{t,s_1,\dots,s_{arity(f)}\})$. 
Choose some $u = u' \in P$. 
If $N' \vdash u=u'$, then every replacement to $f$ is a solution.
So assume $N' \not\vdash u=u'$
Let $w$ is a replacement to $f$ such that $s' := s\{w/f\} = B[w(s_1,\dots,s_{arity(f)})]$ and $N'\{w/f\} \vdash u=u'$.
Let $s' := s\{w/f\}$. 
Assume $s'$ is not $N'$-equivalent to any term in $C$, let $C' := C\cup \Subterms(s')$ and let $A_{N',C'}=(Q, \Sigma, \delta)$.
We know $\delta(s')$ has no outgoing edges: if it did, $s'$ would be $N'$-equivalent to some term in $C$.
By construction,  $A_{N'\cup\{s'=t\},C'}$ is equivalent to calling $merge(\delta(s'), \delta(t))$ on $A_{N',C'}$.
Since $\delta(s')$ has no outgoing edges, calling $merge(\delta(s'), \delta(t))$ on $A_{N',C'}$ cannot induce any more merges.
Therefore, since $\delta(u)$ is not equal to $\delta(u')$, they are not equal after the merge.
So, $s'$ and thus $w(s_1,\dots,s_{arity(f)})$ are $N'$-equivalent to some terms in $C$.

Let $A_{N',C} := (Q, \Sigma, \delta)$. 
For each $q \in Q$, there is a ground term $u_q$ such that $\delta(u_q) = q$. 
Let $Q_{u=u'}$ be the set of states such that $N' \cup \{ B[u_{q'}] = t\} \vdash u = u'$ for each $q' \in Q_{u=u'}$.
Then for each replacement $w$,  $N\{w/f\} \vdash u=u'$ if and only if $\delta(w) \in Q_{u=u'})$.

Let $Q' := \bigcup_{e \in P} Q_e$.
Let $A_{\psi} := (Q, \Sigma \cup \{x_1, \dots, x_{arity(f)}\}, \delta', Q')$ be a tree automaton with accepting states $Q'$.  %\ben{same as above}
For each $x_i$, let $\delta'(x_i) := \delta(s_i)$. 
For all $u \in T(\Sigma)$, let $\delta'(u) := \delta(u)$.
A simple inductive argument will show that $L(A_{\psi})$ are precisely the solutions to $\psi$.

\begin{figure}
\label{ex2fig}
\centering
\begin{minipage}{.5\textwidth}
  \centering
    \begin{tikzpicture}[->,>=stealth',shorten >=1pt,auto,node distance=2.0cm,
                    semithick]
  \tikzstyle{every state}=[]

  \node[state]           (A)                    {1};
  \node[state,accepting] (C) [above right of=A] {3};
  \node[state]           (B) [below right of=C] {2};

  \path (A) edge [bend right]  node[below] {g} (B)
  		(B) edge [bend right]  node[above right] {h} (C)
  		(C) edge [bend right]  node[above left] {g} (A);
  
  \draw[<-,dashed] (B) -- node[above right] {$x$} ++(.8cm,-.8cm);
  \draw[<-] (A) -- node[above left] {$a$} ++(-.8cm,-.8cm);
\end{tikzpicture}
  \label{fig:test1}
\end{minipage}%
\begin{minipage}{.5\textwidth}
  \centering
  \begin{tikzpicture}[->,>=stealth',shorten >=1pt,auto,node distance=2.0cm,
                    semithick]
  \tikzstyle{every state}=[]
  \node[state]           (A)                    {1,2,3};
  \path (A) edge [loop above] node {g,h} (A);
  \draw[<-] (A) -- node[above left] {$a$} ++(-.8cm,-.8cm);
  \end{tikzpicture}
  \label{fig:test2}
\end{minipage}
\caption{Left: The set of solutions to $\psi$ in example \ref{regex2}.
Right: The resulting automaton (without x transition and accepting state) after merging states 1 and 3.}
\end{figure}
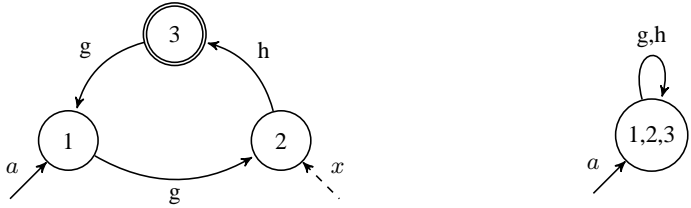 %\ben{need to add "fig-a" and "fig-b" labels to each side}

\begin{example}
\label{regex2}
Let $\psi := (g(h(g(a)))=a \wedge f(g(a))=a) \rightarrow h(g(a))=a$. 
Note that this is a case-2 regular-EUF clause.
If we set $E := \{g(h(g(a)))=a\}$ and $C := \{a,g(a),h(g(a)),g(h(g(a)))\}$, then $A := A_{E,C}$ is the automaton from the left side of figure-\ref{ex2fig} (excluding the accepting state and $x$ transition).
Since the argument of $f$ in $f(g(a))=a$ is $a$ and $A$ parses $a$ to state-2, a transition from $x$ to state-2 is added to $A$.
If we choose a replacement $w$ such that $w(g(a))$ parses to state-3 in $A$, then applying the equation $w(g(a))=a$ merges state-3 with state-1.
This, in turn, forces a merge between the new state and state-2, yielding the automaton on the right side of figure-\ref{ex2fig}.
This automaton parses $h(g(a))$ and $a$ to the same state, so state-3 is an accepting state. %\ben{I'm not sure how to word this better}
This does not occur if $w(g(a))$ parses to state-1 or state-2 in $A$, so they are not accepting states.
So $A$ accepts the replacements $w$ to $f$ such that $\psi\{w/f\}$ is valid.
%\ben{should I include these examples after each case-description, or somewhere else?}
\end{example}

%\ben{do I want to state the above construction inside the proof of a theorem?}
We can summarize the above construction in the following lemma.

\begin{lemma}
\label{in-exp}
The regular SyGuS-EUF problem is in EXPTIME. 
\end{lemma} %\ben{This is true, but can we actually conclude this from the solution to the  non-empty intersection problem?}

The relationship between regular tree languages and the regular SyGuS-EUF problem is quite deep. 
Using the following lemma and the above constructions, we can see that a tree language is regular if and only if it is the set of solutions to 
a regular SyGuS-EUF problem.

\begin{lemma}
\label{aut_to_formula}
Let $A:=(Q, \Sigma \cup \{x_1,\dots,x_k\}, \delta, Q_F)$ be a tree automaton. 
There exists a regular disjunctive formula $\psi_A$ such that $L(A)$ is the set of solutions to $\psi_A$.
\end{lemma}
\begin{proof}
Let $T_Q$ be a subterm-closed set of terms such that for each state $q \in Q$, there is a term $u_q$ such that $\delta(u_q) = q$.
Without loss of generality, assume that each $u_q \in T_Q$ is a subterm of some term in $L(A)$.
Let $\sigma := \{ x_i \mapsto c_i  \mid i \in \{0,\dots,k\} \}$ for some new constants $c_1,\dots,c_k$.
Let  $N_Q :=  \{ g(u_{q_1}\sigma, \dots, u_{q_r}\sigma) = u_{q'}\sigma \mid r \ge 0, g \in \Sigma_r, q_1,\dots,q_r,q' \in Q, \delta(g, q_1, \dots, q_r) = q' \}$ and $P_Q := \{ f(c_1,\dots,c_k) = u_q \mid q \in Q_F\}$.
Finally set $\psi_A := (\bigwedge_{e \in N_Q} e) \rightarrow (\bigvee_{e \in P_Q} e)$. 
Using the construction from theorem \ref{aut_thm}, it is easy to check that the set of solutions to $\psi$  are precisely $L(A)$. \qed
\end{proof}

%It is interesting to note that we can use the above lemma to reduce the $regular SyGuS-EUF$ problem on a given grammar $G$ of possible solutions to $regular SyGuS-EUF$ where any solution is allowed.
%Simply take the automaton $A$ such that $L(A) = L(G)$ and conjunct $\psi_A$ with 
We can also use the above lemma to show that regular SyGuS-EUF is EXPTIME-complete, as we will see below.

\begin{lemma}
\label{exp-hard}
The regular SyGuS-EUF problem is EXPTIME-hard.
\end{lemma}
\begin{proof}
We reduce from the EXPTIME-complete problem of determining whether a set of regular tree automata have languages with a non-empty intersection \cite{veanes1997computational}.
Let $A_1$, \dots, $A_k$ be a set of regular tree automata over some alphabet $\Sigma$.
For each automaton $A_i$, construct the formula $\psi_{A_i}$ as described in lemma \ref{aut_to_formula}.
Let $\phi := \bigwedge_i \psi_{A_i}$.  
Let $f$ be a nullary function symbol to be synthesized, and let $G$ be a grammar such that $L(G) := T(\Sigma)$.
The solutions to the regular SyGuS-EUF problem $(\phi, \Sigma, G, f)$ are the members of the set $\bigcap_i L(A_i)$.
Therefore, $(\phi, \Sigma, G, f)$ has a solution if and only if $\bigcap_i L(A_i)$ is non-empty. \qed
\end{proof}

Using the above lemma and lemma \ref{in-exp}, we can conclude the following theorem.

\begin{theorem}
The regular SyGuS-EUF problem is EXPTIME-complete.
\end{theorem}

In concluding this section, we remark that the case-1 and case-2 restrictions on regular clauses are necessary. 
For lack of space, we exclude the details; the appendix contains an example elaborating on this point.
\smvertspace
\section{Finite-Domain Theories}
\label{sec:fd}
%\subsection{A Generic SyGuS Algorithm}
\smvertspace

In addition to the ``standard'' theories, we also consider a family
of theories that we term {\em finite-domain (FD) theories}.
Formally, an FD theory is a complete theory that admits one domain (up to isomorphism), and whose only domain is finite. 
For example, consider group axioms with a constant $a$ and the statements $\forall x: (x = 0) \vee (x = a) \vee (x = a \cdot a)$ and $a \cdot a \cdot a = 0$.  This is an FD theory, since, up to isomorphism, the only model of this theory is the integers with addition modulo 3. 
%Note that completeness implies that all models of the theory have a finite domain if there is one model with a finite domain (consider the family of sentences that state ``there are not $n$ distinct elements'', which can be expressed only using equality, conjunction, and negation). 
% $M \models \varphi \Leftrightarrow M'
%\models \varphi$ holds for all $M, M' \in \text{Mod}(\mathcal{T})$ and
%all first-order logic formulas $\varphi$. 
% 
Also Boolean logic and the theory of fixed-length bit-vectors without concatenation are FD theories. 
Bit-vector theories with (unrestricted) concatenation allow us to construct arbitrarily many distinct constants and are thus not FD theories. 
%Finite-domain theories include, for example, theories of bit-vectors without the concatenation operation (or a bounded version of concatenation). 

In this section we give a generic algorithm for any complete finite-domain theory for which validity is decidable. 
%\markus{Has this algorithm template been used in some SyGuS papers already? We should not miss such citations.}
Let $\mathcal{T}$ be a such a theory and let $M$ be a model of $\mathcal T$ with a finite domain $\textit{dom}(M)$.
%\ben{I don't think this is enough of a requirement, but I'm leaving it here as a placeholder. I want to be able to talk about all models that satisfy T by showing the values of terms as interpreted under a single model. I think this can be done for fixed-width bitvectors. (I think property I want theories to have is "completeness", meaning that each model satisfies all the same formulas. I'm going to look into this and be careful when i state it.)}   
Assume without loss of generality that for every element $c \in \textit{dom}(M)$ there is a constant $f$ in $M$ such that $f^M = c$.

We consider a SyGuS problem with a correctness specification $\varphi$ in theory $\mathcal{T}$, a function symbol $f$ to synthesize, and a tree grammar $G = (N, S, \mathcal F, P)$ generating the set of candidate expressions. 
Let $\mathbf{a} := a_1,\dots ,a_r$ be the constants occurring in $\varphi$. 
%Let $\mathbf{x} := x_0,x_1,\dots ,x_r$ be the variables (i.e. the uninterpreted symbols) occurring in $\varphi$. \markus{what are uninterpreted symbols? We probably mean just \emph{function symbols}. Also, at some point we reserved $x_i$ for the bound variables of $f$.}
%\markus{what are uninterpreted symbols? We probably mean just \emph{function symbols}. Also, at some point we reserved $x_i$ for the bound variables of $f$.}
The expression $e$ generated by $G$ to replace $f$ can be seen as a function mapping $\mathbf a$ to an element in $\textit{dom}(M)$. 
If the domain of $M$ is finite there are only finitely many candidate functions, but it can be non-trivial to determine which functions can be generated by $G$. 
In the following, we describe an algorithm that iteratively determines the set of functions that can be generated by each non-terminal in the grammar $G$. 

For each $V \in N$, we maintain a set $E_V$ of expressions $e$. 
In each iteration and for each production rule $V\to f(t_1,\dots,t_k)$ for $V$ in $G$, we consider the expressions $f(t_1^*,\dots,t_k^*)$ where $t_i^*:=t_i$ if $t_i$ is an expression (i.e. $t_i\in T_{\mathcal F}$) and $t_i^*\in E_{V'}$ if $t_i$ is a non-terminal $V'$.
Given such an expression $e$, we compute the function table, that is the result of $e\{\mathbf c/\mathbf a\}$ for each $\mathbf c\in\textit{dom}(M)^r$, compare it to the function table of the expressions currently in $E_V$.
Our assumption of decidability of the validity problem for $\mathcal T$ guarantees that this operation is decidable. 
If $e$ represents a new function, we add it to the set $E_V$. 

The algorithm terminates, after an iteration in which no set $E_V$ changed. 
As there are only finitely many functions from $\textit{dom}(M)^r$ to $\textit{dom}(M)$ and the sets $E_V$ grow monotonously, the algorithm eventually terminates.
To determine the answer to the SyGuS problem, we then check whether there is an expression $e$ in $E_S$, for which $\varphi\{e/f\}$ is valid.

\begin{theorem}
Let $\mathcal{T}$ be a complete theory for which validity is decidable and which has a finite-domain model $M$. 
The SyGuS problem for $\mathcal{T}$ and $\mathcal T$-compatible tree grammars is decidable. 
%Let $\varphi$ be a constraint with some uninterpreted function $f$. Let $G$ be a function-call tree grammar representing possible replacements to the function $f$. Then the SyGuS problem for $T$ , $\varphi$, $f$, and $G$ is decidable.
\end{theorem}
%\ben{It's also decidable when synthesizing multiple functions}

\begin{table}[t]
\begin{center}
  \begin{tabular}{| c || c | c | c | }
    \hline
    Iteration\# & $E_S$ & $E_A$ & $E_B$ \\ \hline
    1 & none & $x$ & $y$ \\ \hline
    2 & $x \oplus y$ & $\neg y$ & none \\ \hline
    3 & $\neg y \oplus y \equiv \top$ & none 
        & \scell{  
        $(x \oplus y) \oplus \neg y \equiv \neg x$ \\
        \cancel{$(x \oplus y) \oplus x \equiv y$}
        }
        
    \\ \hline
    4 & \scell{$\neg y \oplus \neg x$ \\ $\cancel{\neg x \oplus x \equiv \top }$ }
        & \scell{$\cancel{\neg \neg x \equiv x}$ \\ none}
        & \scell{$\cancel{\top \oplus \neg y \equiv y}$ \\ 
                 $\cancel{\top \oplus x \equiv \neg x }$ \\
                 none}
        \\ \hline
    5 & none & none 
        & \scell {$\cancel{(\neg y \oplus \neg x) \oplus \neg             y \equiv \neg x}$ \\
                  $\cancel{(\neg y \oplus \neg x) \oplus x \equiv y}$ \\
                  none}
    \\ \hline

  \end{tabular}
\end{center}
\caption{This table shows the expressions added to the sets $E_S$, $E_A$, and $E_B$ when we apply the algorithm to the SyGuS problem in Example~\ref{exmp:getrep}.
%Keep in mind that entries of $F_V$ functions from $x$ and $y$ to some boolean value. 
%The expressions are shown for clarity. 
For readability, we simplify the expressions, indicated by the symbol `$\equiv$'. 
Expressions that are syntactically new, but do not represent a new function are struck out.
When no new function is added, ``none'' is written in the cell.}
\label{fig:extab}
\end{table}

\begin{example}
\label{exmp:getrep}
Consider the SyGuS problem over boolean expressions with the specification $\varphi = x \oplus f$, where $\oplus$ denotes the XOR operation and $f$ is the function symbol to synthesize from the following tree grammar (we use infix operators for readability):
\begin{align*}
S &\rightarrow (A \oplus B) \\
A &\rightarrow \neg B \ | \ x\\
B &\rightarrow (S \oplus A) \ | \ y 
\end{align*}

The grammar generates boolean functions of variables $x$ and $y$ and the updates to $E_A$, $E_B$, and $E_S$ during each iteration of the proposed algorithm are given in Table~\ref{fig:extab}. 
The next step in the algorithm is to determine if any of the three expressions $E_S:=\{x\oplus y,\neg y\oplus y,\neg y\oplus\neg x\}$ make the formula $\varphi\{e/f\}$ valid, which is not the case.

%A graphical representation of the grammar is given in Fig.~\ref{fig:gramgraph} to make the analogy to Bellman-Ford more clear.
%Now suppose we want to calculate the expression for $f = x \oplus \neg y$
%Now suppose we want to calculate the expression for $\bold{f} = x \oplus \neg y$ from cell $(5,S)$. We can trace back through the derivation for $f$ to find:
%
%\begin{align*}
%expr(\neg y \oplus \neg x, S) 
%    &= (expr(\neg y, A) \oplus expr(\neg x, B)) \\
%    &= (\neg expr(y, B) \oplus (expr(x \oplus y, S) \oplus expr(\neg y, A))) \\
%    &= (\neg y  \oplus ((expr(x,A) \oplus expr(y,B)) \oplus \neg expr(y,B)))\\
%    &= (\neg y  \oplus ((x \oplus y) \oplus \neg y))
%\end{align*}
\end{example}

%\begin{figure}[t]
%\centering
%\begin{tikzpicture}
%
%\def \bx {4}
%\def \sx {\bx / 2}
%\def \sy {3.2}
%\node(A){$A$} ;
%\node(B) at ($(A) + (\bx, 0)$){$B$} ; 
%\node(S) at ($(A)+(\sx,\sy)$){$S$} ;  
%\draw(A)[thick] (A) circle (0.3cm);
%\draw[thick] (B) circle (0.3cm);
%\draw[thick] (S) circle (0.3cm);
%\draw[->, bend left, very thick](B) to (A);
%\draw(A)[very thick] to[out=45,in=-90] (\bx/2, \sy/2)--(S);
%\draw(B)[->, very thick] to[out=135,in=-90, very thick] (\bx/2, \sy/2)--(S);
%\draw(S)[->, very thick] to[out=-45,in = 90] (B);
%\draw(A)[very thick] to [out = 0, in = 90](B);
%%\draw (A) to[out=90,in=180] (\bx / 4,\sy / 2)--(\bx / 3,\sy / 2) to[out=0,in=-90](S);
%%\draw[->, thick] (B) to[out=90,in=0] (\bx / 1.33,\sy / 2)--(\bx / 1.5,\sy / 2) to[out=180,in=-90](S);
%
%%labels
%\node(Sxor) at ($(A) + (\bx/2 - 0.5, 0.5 + \sy/2)$){$\bigoplus$};
%\node(NotA) at ($(\bx/2, -0.5)$){$\neg$};
%\node(SAxor) at ($(A) + (\bx - 0.5, 0.5 + \sy/2)$){$\bigoplus$};
%
%
%%\draw[rounded corners=2cm,dashed,->] (A)to[out=90,in=-60] ($heh$);
%%\draw[dashed,->] (A) to [in =-90,out =90](B);
%%\draw[dashed,->] (B) to [in =-180,out =90] ++(2,1) to [out=0,in=-70]($Okay$);
%\end{tikzpicture}
%\caption{A graphical representation of the grammar from Example~\ref{exmp:getrep}.}
%\label{fig:gramgraph}
%\end{figure}

%%%%%%%%%%%%%%%%%%%%%%%%%%%%%%%%%%%%%%%%%%%%%%%%%%%%%%%%%%%%%%%%%%%%%%%%%%%%%%%%%%%%%
\smvertspace
\section{Bit-Vectors}
\label{sec:bv}
\smvertspace

%\markus{We have a serious problem in this section: the proof depends on the syntax. Consider LIA, which only consists of the signature $(0,1,+,-,\leq)$, we cannot construct numbers in this way. Compactness issues in Presburger arithmetic may result in expressiveness/complexity issues in SyGuS-LIA.}

In this section, we show that the SyGuS problem for the theory of bit-vectors is undecidable - even when we restrict the problem to tree grammars.
The proof makes use of the fact that we can construct (bit-)strings with the concatenation operation and  can compare arbitrarily large strings with the equality operation. 
This enables us to encode problem of determining if the languages of CFGs with no $\varepsilon$-transitions have non-empty intersection, which is undecidable \cite{hopcroft1979introduction}.
%The next theorem shows that syntax-guided synthesis for the theory of linear integer arithmetic (SyGuS-LIA) and linear real arithmetic (SyGuS-LRA) is still undecidable when given CFGs that only produce well-formed strings. 
%We reduce the problem of determining the emptiness of the intersection of two context-free grammars to the SyGuS problem. 

\begin{theorem}
	The SyGuS problem for the theory of bit-vectors is undecidable for both the class of context-free grammars and the class of BV-compatible tree grammars.
\end{theorem}

\begin{proof}
	We start with the proof for the class of context-free grammars. 
	Given two context-free grammars $G_1=(N_1,S_1,T_1,R_1)$ and $G_2=(N_1, S_2,T_2,R_2)$, we define a SyGuS problem with a single context-free grammar $G=(N,S,T,R)$ that has a solution iff the intersection of $G_1$ and $G_2$ is not empty.
%	Determining the emptiness of the intersection of two context-free grammars is well known to be undecidable.
	The proof idea is to express the intersection of the two grammars as the equality between two expressions, each generated by one of the grammars.
	The new grammar thus starts with the following production rule:
\[S \rightarrow S_1=S_2\] 

	We then have to translate the grammars $G_1$ and $G_2$ into grammars $G_1'$ and $G_2'$ that produce expressions in the bit vector theory instead of arbitrary strings over their alphabets.
	There is a string produced by both $G_1$ and $G_2$ if and only if the constructed grammars $G_1'$ and $G_2'$ can produce a pair of equal expressions.
	We achieve this by encoding each letter as a bit string of the fixed length $1+\log_2|T_1\cup T_2|$, and by intercalating concatenation operators ($@$) in the production rules: 
	We encode each production rule $(N,P)$ with $P=p_1p_2\dots p_n$ as $(N,P')$ with $P'=p_1' @ p_2' @\dots @ p_n'$, where $p_i'=p_i$ if $p_i\in N$, and otherwise $p_i'$ are the fixed-length encodings of the terminal symbols. 
%	That is, we have to find a class of bit vector expressions for which semantic equivalence is the same as syntactic equivalence. 
%	For LIA, LRA, and DL we transform grammars $G_1$ and $G_2$ into grammars $G_1'$ and $G_2'$ over digits, i.e. with the set of terminals $\{0,\dots,9\}$, while maintaining the equality of each pair of strings produced by the grammars. 
%	We can do so by encoding each terminal in $T_1\cup T_2$ as a sequence of digits of length $1+\log_{10} |T_1\cup T_2|$, where we require the first digit to be 1 to avoid leading zeros. 
	We then define $N=S\dotcup N_1'\dotcup N_2'$,~ $T=\{0,1,@,=\}$, and $R= R_1'\cup R_2' \cup \{(S,S_1'=S_2')\}$. 
	The correctness constraint $\varphi$ of our SyGus problem then only states $\varphi\coloneqq \neg f$, where $f:\mathbb{B}$ is the function symbol, a constant, to synthesize.
	As each character in the alphabets of the context-free grammars was encoded using bit vectors the same length, the comparison of the bit vectors is equivalent to the comparison between the strings of characters of the grammars $G_1$ and $G_2$ and the SyGuS problem has a solution if and only if the intersection of the languages of the context-free grammars $G_1$ and $G_2$ is empty. 

	Note that the context-free grammar $G$ can also be interpreted as a $BV$-compatible tree grammar, where $BV$ is the theory of bit-vectors. 
	Although it is efficiently decidable whether two tree grammars produce a common tree, the expressions produced by the tree-interpretation of $G_1$ and $G_2$ will be equivalent as long as their leaves are equivalent.
	Thus, the equality of the expression trees in the interpretation of the bit vector theory still coincides with the intersection of the given context-free grammars $G_1$ and $G_2$. 
	\qed
\end{proof}

We only used the concatenation operation of the bit-vector theory for the proof. 
That is, SyGuS is even undecidable for fragments of the theory of bit-vectors for which basic decision
problems are easier than the general class; for example, the theory of
\emph{fixed-sized bit-vectors with extraction and composition}~\cite{cyrluk1997efficient} for which
satisfiability of conjunctions of atomic constraints is polynomial-time solvable unlike the general
case which is NP-hard.

\begin{remark}
This proof only relies on the comparison of arbitrarily large values in the underlying logical theory. 
It may thus be possible to extend the proof to other theories involving numbers, such as LIA, LRA, and difference logic. 
The problem here is that these proofs tend to depend on syntactical sugar.
Consider the case of LIA. 
If the signature allows us to use arbitrary integer constants, such as 42, it is simple to translate the proof above into a proof of undecidability of SyGuS for LIA and CFGs. 
For the standard signature of LIA, however, which just includes the integer constants 0 and 1 (larger integers can then be expressed as the repeated addition of the constant 1) the proof scheme above does not apply. 
\end{remark}

%%%%%%%%%%%%%%%%%%%%%%%%%%%%%%%%%%%%%%%%%%%%%%%%%%%%%%%%%%%%%%%%%%%%%%%%%%%%%%%%%%%%%
\smvertspace
\section{Other Background Theories}
\label{sec:misc}
\smvertspace

In this section, we remark on the decidability for some
other classes of SyGuS problems. These results are 
straightforward, but the classes do occur in practice, and so
they are worth mentioning.

%\noindent
%{\bf Linear real arithmetic (LRA) with arbitrary affine expressions.}
\subsubsection*{Linear real arithmetic (LRA) with arbitrary affine expressions.}
Consider the family of SyGuS problems where:
\begin{itemize} 
\item [i)] the specification $\varphi$ is a
Boolean combination of linear constraints over real-valued variables
$\vec{x} := x_1, x_2, \ldots, x_n$ 
and applications of the function $f$ to be synthesized. For
simplicity, we assume a single function $f$ of arity $n$; the
arguments below generalize. 
\item [ii)]
The grammar $G$ is the one generating {\em arbitrary affine expressions} 
over $\vec{x}$ to replace for applications of $f$. 
Thus, the application $f(\vec{t})$,
where $\vec{t} := t_1, t_2, \ldots, t_n$ is a vector of LRA terms, is
replaced by an expression of the form $a_0 + \sum_{i=1}^n a_i t_i$. 
\end{itemize}
Thus, for a fixed set of variables $\vec{x}$ there is a fixed grammar
for all formulas $\varphi$.

This case commonly arises in invariant synthesis when the invariant is
hypothesized to be an affine constraint over terms in a program. 
In this case, the solution of the SyGuS problem reduces to solving the
$\exists \forall$ SMT problem 
$$
\exists a_0,a_1,\ldots,a_n \, . \, \forall x_1 x_2 \ldots x_n \, . \,
\bigl( \varphi [f(\vec{t}) / a_0 + \sum_{i=1}^n a_i t_i] \bigr)
$$
which reduces to a formula with first-order quantification over real
variables. Since the theory of linear real arithmetic admits
quantifier elimination, the problem is solvable using any of a number
of quantifier elimination techniques, including classic methods such
as Fourier-Motzkin elimination~\cite{dantzig-jct73} 
and the method of Ferrante and Rackhoff~\cite{ferrante-sicomp75},
as well as more recent methods for solving exists-forall SMT problems
(e.g.,~\cite{dutertre2015solving}). 

%- LIA/LRA is decidable for unrestricted expressions
%\cite{reynolds-cav15} \ben{This is the citation to the tinelli,Barret,Cav'15 paper.}
This decidability result continues to hold for grammars that generate bounded-depth
conditional affine expressions. However, the case of unbounded-depth
conditional affine expressions is still open, to our knowledge.

A similar reduction, for the case of affine expressions, can be
performed for linear arithmetic over the integers (LIA), requiring
quantifier elimination for Presburger arithmetic. Thus, this case is
also decidable.

%\noindent
%{\bf Finite-precision bit-vector arithmetic (BV) with arbitrary bit-vector
%functions.}
\subsubsection*{Finite-precision bit-vector arithmetic (BV) with arbitrary bit-vector
functions.}
Consider the family of SyGuS problems where:
\vspace{-2pt}
\begin{itemize}
\item [i)] the specification $\varphi$ is an
arbitrary formula in the quantifier-free theory of finite-precision
bit-vector arithmetic~\cite{barrett-smtbookch09,BarFT-SMTLIB} over
a collection of $k$ bit-vector variables whose cumulative bit-width is
$w$. Let $f$ be a bit-vector function to be synthesized with output
bit-width $m$.
\item [ii)]
The grammar $G$ is the one generating {\em arbitrary bit-vector
expressions} over these $k$ variables, using all the operators defined
in the theory. In other words, $G$ imposes no major syntactic
restriction on the form of the bit-vector function $f$.
\end{itemize}
Thus, for a fixed set of bit-vector variables there is a fixed grammar
for all formulas $\varphi$.

This class of SyGuS problems has been studied as the synthesis of
``bitvector programs'' (in applications such as code optimization
and program deobfuscation) from
components (bit-vector operators and
constants)~\cite{jha-icse10,gulwani-pldi11}.  
It is easy to see that this class is decidable.
A simplistic (but not very efficient) way to solve it is to enumerate
all $2^{m2^w}$ possible semantically-distinct bitvector functions over
the $k$ variables and check, via an SMT query, whether each, when substituted
for $f$ will make the resulting formula valid.

%%%%%%%%%%%%%%%%%%%%%%%%%%%%%%%%%%%%%%%%%%%%%%%%%%%%%%%%%%%%%%%%%%%%%%%%%%%%%%%%%%%%%
\smvertspace
\section{Discussion}
\label{sec:discuss}
\smvertspace

In this paper, we have presented a first theoretical analysis of the SyGuS problem,
focusing on its decidability for various combinations of logical theories and grammars.
The main results of the paper are summarized in Table~\ref{tbl:main-results-summary},
augmented by the decidability of the simple but common SyGuS classes
described in Section~\ref{sec:misc}. We conclude with a few remarks about
the results, connections between them, and their relevance in practice.

Consider the theory of finite-precision bit-vector arithmetic (BV). 
We have seen in Section~\ref{sec:misc} that the SyGuS problem is decidable when the logical formula is an arbitrary BV formula and the grammar allows the function to be replaced by any bit-vector function over the constants in the formula. %, but not by something more restrictive. 
However, we have also seen that the SyGuS problem is undecidable when an arbitrary context-free grammar can be used to
restrict the space of bit-vector functions to be synthesized (see Section~\ref{sec:bv}). 
%In the latter case, we are allowing many more restrictive grammars. 
These results may seem to contradict the intuition (stated in Section~\ref{sec:intro}) that syntax guidance 
restricts the search space for synthesis and thus makes the problem easier to solve. 
We thus have to be careful which classes of grammars we pick to restrict SyGuS problems. % and pick grammars that represent restrictions that are \emph{useful} for the decision procedures. 
%One possible interpretation of the results 
%is that it is not just sufficient to allow (arbitrarily) restrictive grammars but also to pick the {\em right restrictions}. 
%In less formal terms, ``too much freedom'' to choose a restrictive grammar can make the SyGuS problem harder.

%In this paper we have seen that many SyGuS problems which are decidable when a function can be replaced by any expression become undecidable when the replacement expressions are constrained to fit a grammar. This may not be so surprising since so many problems for formal grammars are undecidable. In cases when the synthesized expression must fit a certain structure, the loss of decidability is understandable. When any replacement expression would be allowed, however, it may not be clear why a grammar should be provided at all. It is important to realize, though, that making a problem undecidable by allowing for the use of grammars does not mean that the problem will be more difficult in practice.

%\subsection{Future Work}

For future work, it would be good to study the LIA and LRA background theories
in more detail. In particular, we would like to determine if these
theories are decidable when grammars are provided, and whether the use
of conditionals without bounding expression tree depth affects the decidability.  
Further, for SyGuS classes that are decidable, it would be useful to perform
a more fine-grained characterization of problem complexity, especially with
regard to special classes of grammars.

% Only in final version:
%\paragraph{Acknowledgements.} This work was supported in part by NSF Expeditions project CCF-1139138 and by NSF awards \#1329759 and \#1139138.
%\vspace{-.2cm}

%%%%%%%%%%%%%%%%%%%%%%%%%%%%%%%%%%%%%%%%%%%%%%%%%%%%%%%%%%%%%%%%%%%%%%%%%%%%%%%%%%%%%
\bibliographystyle{plain}
\bibliography{bibliography}

\begin{thebibliography}{10}

\bibitem{ScenariosHVC2014}
R.~Alur, M.~Martin, M.~Raghothaman, C.~Stergiou, S.~Tripakis, and A.~Udupa.
\newblock {Synthesizing Finite-state Protocols from Scenarios and
  Requirements}.
\newblock In {\em Proceedings of Haifa Verification Conference (HVC)}, volume
  8855 of {\em LNCS}. Springer, 2014.

\bibitem{alur-fmcad13}
Rajeev Alur, Rastislav Bodik, Garvit Juniwal, Milo M.~K. Martin, Mukund
  Raghothaman, Sanjit~A. Seshia, Rishabh Singh, Armando Solar-Lezama, Emina
  Torlak, and Abhishek Udupa.
\newblock Syntax-guided synthesis.
\newblock In {\em Proceedings of Formal Methods in Computer-Aided Design
  (FMCAD)}, pages 1--17, October 2013.

\bibitem{baader1999term}
Franz Baader and Tobias Nipkow.
\newblock {\em Term rewriting and all that}.
\newblock Cambridge university press, 1999.

\bibitem{BarFT-SMTLIB}
Clark Barrett, Pascal Fontaine, and Cesare Tinelli.
\newblock {The Satisfiability Modulo Theories Library (SMT-LIB)}.
\newblock {\tt www.SMT-LIB.org}, 2016.

\bibitem{barrett-smtbookch09}
Clark Barrett, Roberto Sebastiani, Sanjit~A. Seshia, and Cesare Tinelli.
\newblock Satisfiability modulo theories.
\newblock In Armin Biere, Hans van Maaren, and Toby Walsh, editors, {\em
  Handbook of Satisfiability}, volume~4, chapter~8. IOS Press, 2009.

\bibitem{bryant1999exploiting}
Randal~E Bryant, Steven German, and Miroslav~N Velev.
\newblock Exploiting positive equality in a logic of equality with
  uninterpreted functions.
\newblock In {\em Proceedings of Computer Aided Verification (CAV)}, pages
  470--482. Springer, 1999.

\bibitem{colon-cav03}
Michael Col{\'o}n, Sriram Sankaranarayanan, and Henny Sipma.
\newblock Linear invariant generation using non-linear constraint solving.
\newblock In {\em Proceedings of Computer Aided Verification (CAV)}, pages
  420--432, 2003.

\bibitem{tata2007}
H.~Comon, M.~Dauchet, R.~Gilleron, C.~L\"oding, F.~Jacquemard, D.~Lugiez,
  S.~Tison, and M.~Tommasi.
\newblock Tree automata techniques and applications.
\newblock Available on: \url{http://www.grappa.univ-lille3.fr/tata}, 2007.
\newblock release October, 12th 2007.

\bibitem{cyrluk1997efficient}
David Cyrluk, Oliver M{\"o}ller, and Harald Rue{\ss}.
\newblock An efficient decision procedure for the theory of fixed-sized
  bit-vectors.
\newblock In {\em Proceedings of Computer Aided Verification (CAV)}, pages
  60--71. Springer, 1997.

\bibitem{dantzig-jct73}
G.~B. Dantzig and B.~C. Eaves.
\newblock {Fourier-Motzkin} elimination and its dual.
\newblock {\em Journal of Combinatorial Theory A}, 14:288--297, 1973.

\bibitem{degtyarev1996undecidability}
Anatoli Degtyarev and Andrei Voronkov.
\newblock The undecidability of simultaneous rigid e-unification.
\newblock {\em Theoretical Computer Science}, 166(1):291--300, 1996.

\bibitem{dutertre2015solving}
Bruno Dutertre.
\newblock Solving exists/forall problems with {Yices}.
\newblock In {\em Proceedings of the International Workshop on Satisfiability
  Modulo Theories (SMT)}, 2015.

\bibitem{ferrante-sicomp75}
Jeanne Ferrante and Charles Rackoff.
\newblock A decision procedure for the first order theory of real addition with
  order.
\newblock {\em {SIAM} Journal of Computing}, 4(1):69--76, 1975.

\bibitem{gulwani-pldi11}
Sumit Gulwani, Susmit Jha, Ashish Tiwari, and Ramarathnam Venkatesan.
\newblock Synthesis of loop-free programs.
\newblock {\em SIGPLAN Notices}, 46:62--73, June 2011.

\bibitem{hopcroft1979introduction}
John~E Hopcroft and Jeffrey~D Ullman.
\newblock {\em Introduction to automata theory, languages, and computation}.
\newblock Pearson Education India, 1979.

\bibitem{jha-icse10}
Susmit Jha, Sumit Gulwani, Sanjit~A. Seshia, and Ashish Tiwari.
\newblock Oracle-guided component-based program synthesis.
\newblock In {\em Proceedings of the 32Nd ACM/IEEE International Conference on
  Software Engineering (ICSE)}, pages 215--224, 2010.

\bibitem{kozen1977complexity}
Dexter Kozen.
\newblock Complexity of finitely presented algebras.
\newblock In {\em Proceedings of the ninth annual ACM symposium on Theory of
  computing}, pages 164--177. ACM, 1977.

\bibitem{kozen1992myhill}
Dexter Kozen.
\newblock On the myhill-nerode theorem for trees.
\newblock {\em Bull. Europ. Assoc. Theor. Comput. Sci}, 47:170--173, 1992.

\bibitem{Kroening2008}
Daniel Kroening and Ofer Strichman.
\newblock {\em Equality Logic and Uninterpreted Functions}, pages 59--80.
\newblock Springer Berlin Heidelberg, Berlin, Heidelberg, 2008.

\bibitem{MannaWaldinger80}
Z.~Manna and R.~Waldinger.
\newblock A deductive approach to program synthesis.
\newblock {\em {ACM TOPLAS}}, 2(1):90--121, 1980.

\bibitem{sketching:pldi05}
Armando Solar-Lezama, Rodric Rabbah, Rastislav Bod\'{\i}k, and Kemal Ebcioglu.
\newblock Programming by sketching for bit-streaming programs.
\newblock In {\em Proceedings of the 2005 ACM SIGPLAN Conference on Programming
  Language Design and Implementation (PLDI)}, pages 281--294, 2005.

\bibitem{sketching:asplos06}
Armando Solar-Lezama, Liviu Tancau, Rastislav Bod\'{\i}k, Sanjit~A. Seshia, and
  Vijay Saraswat.
\newblock Combinatorial sketching for finite programs.
\newblock In {\em ASPLOS}, pages 404--415, 2006.

\bibitem{udupa-pldi13}
Abhishek Udupa, Arun Raghavan, Jyotirmoy~V. Deshmukh, Sela Mador-Haim,
  Milo~M.K. Martin, and Rajeev Alur.
\newblock {\textsc{Transit}}: Specifying protocols with concolic snippets.
\newblock In {\em Proceedings of the $34^{th}$ ACM SIGPLAN conference on
  Programming Language Design and Implementation}, pages 287--296, 2013.

\bibitem{veanes1997computational}
Margus Veanes.
\newblock On computational complexity of basic decision problems of finite tree
  automata.
\newblock Technical report, UPMAIL Technical Report 133, Uppsala University,
  Computing Science Department, 1997.

\end{thebibliography}

\newpage
\appendix

\section{Omitted details from Sec.~\ref{sec:reg-euf}}

At the end of Sec.~\ref{sec:reg-euf}, we remarked that the case-1 and case-2 restrictions on regular clauses are necessary. 
The following example gives a clause that includes one positive and one negative equation in which an $f$ appears. 
The set of solutions to the corresponding SyGuS-EUF problem is not a regular tree language.
more specifically:

Let $\Sigma := \{ g\!:\!1, g'\!:\!1, h\!:\!1, a\!:\!0, b\!:\!0, c\!:\!0 \}$ be a ranked alphabet, and let $f$ be a unary function symbol to be synthesized. %\!=\!\ben{can we make the spacing around $\!:\!$ and $\!=\!$ smaller?}
Let $N := \{ f(a)\!=\!b, g(a)\!=\!a, g'(a)\!=\!a, h(a)\!=\!b, h(b)\!=\!c, g(c)\!=\!c, g'(c)\!=\!c\}$ and $\phi := (\bigwedge_e \in P e) \rightarrow f(b)\!=\!c$.  
Define $G$ to be the tree grammar with start symbol $S$ and the following rules: $S \rightarrow g(S) | g'(S) | h(A)$ and $A \rightarrow g(A) | g'(A) | h(A) | x$.
We will show that the set of solutions to the regular SyGuS-EUF problem $(\phi, \Sigma, G, f)$ is not a regular tree language.

Let $w(x) \in L(G)$ be a replacement to $f$ and $E' := E\{w/f\}$.
By the rules of $G$, there must be a context $B$ and a term $t(x)$ over the alphabet $\{g\!:\!1, g'\!:\!1\}$ such that $w(x) \!=\! B[h(t(x))]$.
We can see that $b =_{E'} w(a) =_{E'} B[h(t(a))] =_{E'} B[h(a)] =_{E'} B[b]$.
Also, $w(b) =_{E'} c \Leftrightarrow h(t(b)) =_{E'} c \Leftrightarrow  t(b) =_{E'} b$.
The terms $t(x)$ such that $t(b)\!=\!b$  are precisely those of the form  $B[B[\dots B[x]\dots]]$. %\ben{need more explanation here?}
Therefore, the set of solutions to the above regular SyGuS-EUF is $L := \{ B[h(B[B[\dots B[x]\dots]])] \mid B \mbox{ is any context over } \{g\!:\!1. g'\!:\!1\}\}$.

We now use the Myhill-Nerode theorem for regular tree languages \cite{kozen1992myhill}, stated below: 
\begin{theorem}[Myhill-Nerode theorem for regular tree languages~\cite{tata2007}]
\label{thm:MNregulartreelanguages}
Given a tree language $L$ over ranked alphabet $\Sigma$, we define $s \equiv_L t$ if $C[s] \Leftrightarrow C[t]$ for each context $C$ and terms $s$ and $t$ over $\Sigma$. 
The following are equivalent:
\begin{enumerate}
\item $L$ is regular
\item $\equiv_L$ has finitely many equivalence classes
\item $L$ is accepted by a rational tree automaton. 
\end{enumerate}	
\end{theorem}

Using this theorem, it is easy to check that $L$ is not regular.

\end{document}